\documentclass[runningheads]{llncs}

\usepackage{amssymb}
\usepackage{graphicx}
\usepackage{amsmath}
\usepackage{url}
\usepackage{times}
\usepackage{enumitem}
\usepackage{mymacros}

\def\cluster#1{\ensuremath{[#1]_c}}
\def\pcomp{\mi{PComp}}
\def\tcomp{\mi{TComp}}

\begin{document}

\mainmatter  

\title{Markings in Perpetual Free-Choice Nets Are Fully Characterized by Their Enabled Transitions}

\titlerunning{Markings in Perpetual Free-Choice Nets}

%
%
\author{Wil M.P. van der Aalst}

\authorrunning{Wil van der Aalst}


\institute{Process and Data Science (PADS), RWTH Aachen University, Germany.\\
\email{wvdaalst@pads.rwth-aachen.de}}

%
%
\maketitle

\begin{abstract}
A marked Petri net is \emph{lucent} if there are no two different reachable markings enabling the same set of transitions, i.e.,
states are fully characterized by the transitions they enable.
This paper explores the class of marked Petri nets that are lucent
and proves that \emph{perpetual marked free-choice nets} are lucent.
Perpetual free-choice nets are free-choice Petri nets that are live and bounded and have a home cluster, i.e.,
there is a cluster such that from any reachable state there is a reachable state marking the places of this cluster.
A home cluster in a perpetual net serves as a ``regeneration point'' of the process, e.g.,
to start a new process instance (case, job, cycle, etc.).
Many ``well-behaved'' process models fall into this class.
For example, the class of short-circuited sound workflow nets is perpetual.
Also, the class of processes satisfying the conditions of the $\alpha$ algorithm for process discovery falls into this category.
This paper shows that the states in a perpetual marked free-choice net are fully characterized
by the transitions they enable, i.e., these process models are lucent.
Having a one-to-one correspondence between the actions that can happen and
the state of the process, is valuable in a variety of application domains.
The full characterization of markings in terms of enabled transitions
makes perpetual free-choice nets interesting for workflow analysis and process mining.
In fact, we anticipate new verification, process discovery, and conformance checking techniques for the subclasses identified.
\end{abstract}

\section{Introduction}
\label{sec:intro}

Structure theory is a branch in Petri nets \cite{mbp-aal-stahl-2011,murata,reisig-book-2013,lopn1,lopn2} that asks what behavioral properties
can be derived from it structural properties \cite{bestfcn,structure-theory-ToPNoC-advanced-course2010,deselesparza}.
Many different subclasses have been studied.
Examples include state machines, marked graphs, free-choice nets, asymmetric choice nets, and nets without TP and PT handles.
Structure theory also studies local structures such as traps and siphons that may reveal information about the behavior of the Petri net
and includes linear algebraic characterizations of behavior involving the matrix equation or invariants \cite{structure-theory-ToPNoC-advanced-course2010,deselesparza,murata}.

In this paper, we focus on the following fairly general question:
\emph{What is the class of Petri nets for which each marking is uniquely identified by the set of enabled transitions?}
We call such nets \emph{lucent}. A lucent marked Petri net cannot have two different reachable markings that enable the same set of transitions.
\begin{figure}[t]
\centering
\includegraphics[width=4cm]{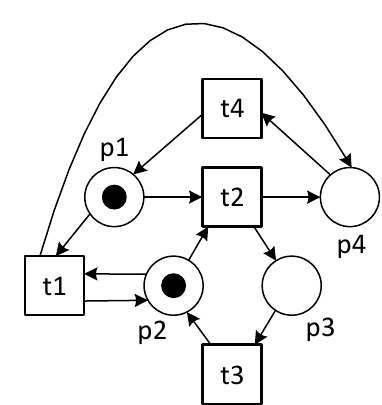}
\caption{A perpetual marked free-choice net (i.e., live, bounded, and having a home cluster) that is lucent (each reachable marking is unique in terms of the transitions it enables).}
\label{fig-alternating}
\end{figure}

Consider, for example, the Petri net shown in Figure~\ref{fig-alternating}.
There are four reachable markings.
Marking $[p1,p2]$ enables $\{t1,t2\}$.
Marking $[p1,p3]$ enables $\{t3\}$.
Marking $[p2,p4]$ enables $\{t4\}$.
Marking $[p3,p4]$ enables $\{t3,t4\}$.
Hence, the marked net is lucent, because each of the four markings is uniquely identified by a particular set of enabled transitions.
The Petri net shown in Figure~\ref{fig-intro-nfc} is not lucent.
After firing either transition $t1$ or $t2$ only $t3$ is enabled, i.e.,
the two corresponding $[p2,p5]$ and $[p2,p6]$ markings enable the same set of transitions.
The choice between $t4$ and $t5$ is controlled by a token in $p5$ or $p6$, and this state information is not ``visible'' when only $t3$ is enabled.
As illustrated by Figure~\ref{fig-intro-nfc}, it is easy to construct non-free-choice nets that are not lucent.
Moreover, unbounded Petri nets cannot be lucent. These examples trigger the question: Which classes of marked Petri nets are guaranteed to be lucent?

In this paper, we will show that \emph{perpetual marked free-choice nets} are always lucent.
These nets are live and bounded and also have a so-called \emph{home cluster}.
A home cluster serves as a ``regeneration point'', i.e., a state where all tokens mark a single cluster.
The property does not hold in general. Liveness, boundedness, the existence of a home cluster, and the free-choice requirement are all needed.
We will provide various counterexamples illustrating that dropping one of the requirements is not possible.

Free-choice nets are well studied \cite{bdefcn,structure-theory-ToPNoC-advanced-course2010,Esparza98TCS,thiavoss}. The definite book on the structure theory of free-choice nets is \cite{deselesparza}.
Also, see \cite{structure-theory-ToPNoC-advanced-course2010} for pointers to literature.
Therefore, it is surprising that the question whether markings are uniquely identified by the set of enabled transitions (i.e., lucency)
has not been explored in literature. Most related to the results presented in this paper is the work on the so-called \emph{blocking theorem}
\cite{blocking-theorem-gaujala2003,blocking-theorem-wehler2010}.
Blocking markings are reachable markings which enable transitions from only a single cluster. Removing the cluster yields a dead marking.
Figure~\ref{fig-alternating} has three blocking markings ($[p1,p2]$, $[p1,p3]$, and $[p2,p4]$). The blocking theorem states that in a bounded and live free-choice net each cluster has a unique blocking marking.
We will use this result, but prove a much more general property.
\emph{Note that we do not look at a single cluster and do not limit ourselves to blocking markings.}
We consider all markings including states (partially) marking multiple clusters.
\begin{figure}[t]
\centering
\includegraphics[width=7cm]{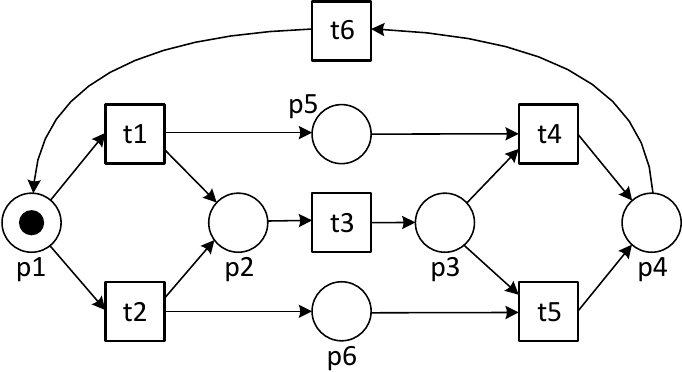}
\caption{A perpetual marked non-free-choice net that is not lucent because there are two reachable markings ($[p2,p5]$ and $[p2,p6]$) enabling the same set of transitions ($\{t3\}$).}
\label{fig-intro-nfc}
\end{figure}

We expect that the theoretical results presented in this paper will enable new analysis techniques in
related fields such as \emph{business process management} \cite{BPM_book_Marlon_Marcello_Jan_Hajo_2013}, \emph{workflow management} \cite{patterns-book-2016}, and \emph{process mining} \cite{process-mining-book-2016}.
Lucency is a natural assumption in many application domains and should be exploited.
For example, the worklists of a workflow management system show the set of enabled actions.
Hence, lucency allows us to reason about the internal state of the system in terms of the actions it allows.
We also anticipate that lucency can be exploited in workflow verification, process discovery, and conformance checking \cite{wires-replay}.
Event logs used in process mining only reveal the actions performed and not the internal state \cite{rcis2013-keynote-the-quest-for-the-right-model,process-mining-book-2016}.
Moreover, the class of perpetual marked free-choice nets considered in this paper is quite large and highly relevant in many application domains.
The existence of a ``regeneration point'' (home cluster) is quite general, and liveness and boundedness (or soundness) are often desirable.
For example, the class of short-circuited sound workflow nets is perpetual.
Processes that are cyclic, often have a home cluster.
Non-cyclic process often have a clear start and end state and can be short-circuited thus introducing a home cluster.
For example, the representational bias of the $\alpha$ algorithm (i.e., the class of process models for which rediscovery is guaranteed)
corresponds to a subclass of perpetual marked free-choice nets \cite{aal_min_TKDE}.

The remainder of this paper is organized as follows.
Section~\ref{sec:prelim} introduces preliminaries and known results (e.g., the blocking theorem).
Section~\ref{sec:lucency} defines lucency as a (desirable) behavioral property of marked Petri nets.
Section~\ref{sec:charmark} introduces perpetual nets and important notions like partial P-covers and local safeness.
These are used to prove the main theorem of this paper showing that markings are unique in terms of the transitions they enable.
Section~\ref{sec:concl} concludes the paper.

\section{Preliminaries}
\label{sec:prelim}

This section introduces basic concepts related to Petri nets, subclasses of nets (e.g., free-choice nets and workflow nets),
and blocking markings.

\subsection{Petri Nets}
\label{sec:petrinets}

Multisets are used to represent the state of a Petri net.
$\bag(A)$ is the set of all multisets over some set $A$.
For some multiset $b\in \bag(A)$, $b(a)$ denotes the number of times element $a\in A$ appears in $b$.
Some examples: $b_1 = [~]$, $b_2 = [x,x,y]$, $b_3 = [x,y,z]$, $b_4 = [x,x,y,x,y,z]$, and $b_5 = [x^3,y^2,z]$
are multisets over $A=\{x,y,z\}$.
$b_1$ is the empty multiset, $b_2$ and $b_3$ both consist of three elements, and
$b_4 = b_5$, i.e., the ordering of elements is irrelevant and a more compact notation may be used for repeating elements.

The standard set operators can be extended to multisets, e.g., $x\in b_2$, $b_2 \bplus b_3 = b_4$, $b_5 \setminus b_2 = b_3$, $|b_5|=6$, etc.
$\{a \in b\}$ denotes the set with all elements $a$ for which $b(a) \geq 1$.
$[f(a) \mid a \in b]$ denotes the multiset where element $f(a)$ appears $\sum_{x \in b \mid f(x)=f(a)}\ b(x)$ times.

$\sigma = \langle a_1,a_2, \ldots, a_n\rangle \in X^*$ denotes a sequence over $X$ of length $n$.
$\langle~\rangle$ is the empty sequence. Sequences can be concatenated using ``$\cdot$'', e.g., $\langle a,b \rangle \cdot \langle b,a \rangle = \langle a,b,b,a \rangle$.
It is also possible to project sequences: $\langle a,b,b,a,c,d\rangle \tproj_{\{a,c\}} = \langle a,a,c\rangle$.

\begin{definition}[Petri Net]\label{def:pn}
A Petri net is a tuple $N=(P,T,F)$ with $P$ the non-empty set of places, $T$ the non-empty set of transitions such that
$P \cap T = \emptyset$, and $F\subseteq (P \times T) \cup (T \times P)$ the flow relation such that the graph $(P \cup T, F)$ is connected.
\end{definition}

\begin{definition}[Marking]\label{def:mrk}
Let $N=(P,T,F)$ be a Petri net.
A marking $M$ is a multiset of places, i.e., $M \in \bag(P)$.
$(N,M)$ is a marked net.
\end{definition}

For a subset of places $X \subseteq P$: $M \tproj_X = [p \in M \mid p \in X]$ is the marking projected on this subset.

A Petri net $N=(P,T,F)$ defines a directed graph with nodes $P\cup T$ and edges $F$.
For any $x\in P\cup T$, ${\prenet{N}{x}} = \{y\mid (y,x)\in F\}$ denotes the set of input nodes and
${\postnet{N}{x}} = \{y\mid (x,y)\in F\}$ denotes the set of output nodes.
The notation can be generalized to sets: ${\prenet{N}{X}}=\{y\mid \exists_{x\in X} \ (y,x)\in F\}$ and ${\postnet{N}{X}} = \{y\mid \exists_{x\in X} \ (x,y)\in F\}$
for any $X \subseteq P\cup T$.
We drop the superscript $N$ if it is clear from the context.

A transition $t \in T$ is \emph{enabled} in marking $M$ of net $N$, denoted as $(N,M)[t\rangle$, if each of its input places ${\pre{t}}$ contains at least one token.
$\mi{en}(N,M) = \{ t \in T \mid (N,M)[t\rangle \}$ is the set of enabled transitions.

An enabled transition $t$ may \emph{fire}, i.e., one token is removed from each of the input places ${\pre{t}}$ and
one token is produced for each of the output places ${\post{t}}$.
Formally: $M' = (M\bmin {\pre{t}})\bplus {\post{t}}$ is the marking resulting from firing enabled transition $t$ in marking $M$ of Petri net $N$.
$(N,M)[t\rangle (N,M')$ denotes that $t$ is enabled in $M$ and firing $t$ results in marking $M'$.

Let $\sigma = \langle t_1,t_2, \ldots, t_n \rangle \in T^*$ be a sequence of transitions.
$(N,M)[\sigma\rangle (N,M')$ denotes that there is a set of markings $M_0, M_1, \ldots, M_n$ ($n \geq 0$)
such that $M_0 = M$, $M_n = M'$, and $(N,M_i)[t_{i+1}\rangle (N,M_{i+1})$ for $0 \leq i < n$.
A marking $M'$ is \emph{reachable} from $M$ if there exists a \emph{firing sequence} $\sigma$ such that $(N,M)[\sigma\rangle (N,M')$.
$R(N,M) = \{ M' \in \bag(P) \mid \exists_{\sigma \in T^*} \ (N,M)[\sigma\rangle (N,M') \}$ is the set of all reachable markings.

Figure~\ref{fig-not-home} shows a marked Petri net having 8 places and 7 transitions. Transitions $t3$ and $t6$ are enabled in the initial marking $M=[p3,p6]$.
$R(N,M) = \{ [p3,p6], [p6,p7],\allowbreak [p3,p8], [p7,p8], [p1,p2], [p3,p4], [p5,p6], [p4,p7], [p5,p8]\}$.
For example, the firing sequence $\langle t3,t6,t7 \rangle$ leads to marking $[p1,p2]$, i.e., $(N,[p3,p6])[\langle t3,t6,t7 \rangle \rangle (N,[p1,p2])$.
\begin{figure}[thbp]
\centering
\includegraphics[width=6cm]{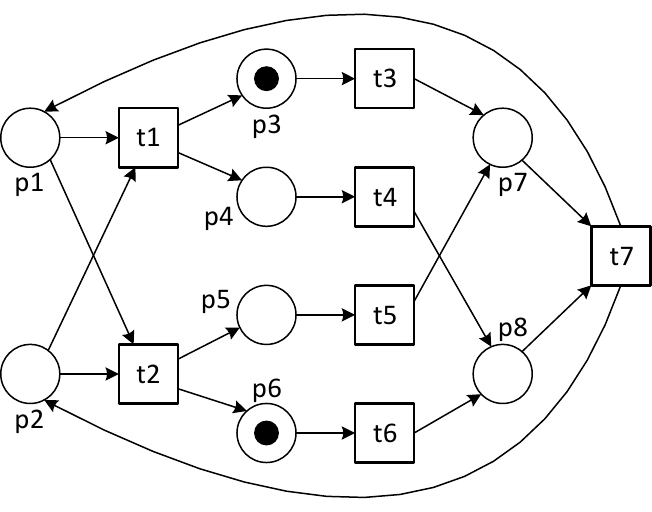}
\caption{A perpetual marked free-choice net \cite{deselesparza}. The net is live, bounded, and has so-called ``home clusters'' (e.g., $\{p7,p8,t7\}$). The net is also lucent.}
\label{fig-not-home}
\end{figure}

A \emph{path} in a Petri net $N=(P,T,F)$ is a sequence of nodes $\rho = \langle x_1,x_2, \ldots ,x_n \rangle$ such that $(x_i,x_{i+1}) \in F$ for $1 \leq i < n$.
$\rho$ is an elementary path if $x_i \neq x_j$ for $1 \leq i < j \leq n$.

Next, we define a few, often desirable, behavioral properties: liveness, boundedness, and the presence of (particular) home markings.

\begin{definition}[Liveness and Boundedness]\label{def:lb}
A marked net $(N,M)$ is \emph{live} if for every reachable marking $M' \in R(N,M)$ and every transition $t\in T$ there exists a marking $M'' \in R(N,M')$ that enables $t$.
A marked net $(N,M)$ is $k$-bounded if for every reachable marking $M' \in R(N,M)$ and every $p \in P$: $M'(p) \leq k$.
A marked net $(N,M)$ is \emph{bounded} if there exists a $k$ such that $(N,M)$ is $k$-bounded.
A 1-bounded marked net is called \emph{safe}.
A Petri net $N$ is \emph{structurally bounded} if $(N,M)$ is bounded for any marking $M$.
A Petri net $N$ is \emph{structurally live} if there exists a marking $M$ such that $(N,M)$ is live.
A Petri net $N$ is \emph{well-formed} if there exists a marking $M$ such that $(N,M)$ is live and bounded.
\end{definition}

The marked Petri net shown in Figure~\ref{fig-not-home} is live and safe. Hence, it is also well-formed.

\begin{definition}[Home Marking]\label{def:hm}
Let $(N,M)$ be a marked net.
A marking $M_H$ is a home marking if for every reachable marking $M' \in R(N,M)$: $M_H \in R(N,M')$.
$(N,M)$ is cyclic if $M$ is a home marking.
\end{definition}

The marked Petri net shown in Figure~\ref{fig-not-home} has 8 home markings: $\{ [p6,p7],\allowbreak [p3,p8],\allowbreak [p7,p8],\allowbreak [p1,p2],\allowbreak [p3,p4],\allowbreak [p5,p6],\allowbreak [p4,p7],\allowbreak [p5,p8]\}$.
However, the net is not cyclic because $[p3,p6]$ is not a home marking.

\subsection{Subclasses of Petri Nets}
\label{sec:subclass}

For particular subclasses of Petri net there is a relationship between structural properties and behavioral properties like liveness and boundedness \cite{structure-theory-ToPNoC-advanced-course2010}.
In this paper, we focus on free-choice nets \cite{deselesparza}.

\begin{definition}[P-net, T-net, and Free-choice Net]\label{def:stfcn}
Let $N=(P,T,F)$ be a Petri net.
$N$ is an P-net (also called a state machine) if $\card{{\pre{t}}} = \card{{\post{t}}} = 1$ for any $t \in T$.
$N$ is a T-net (also called a marked graph) if $\card{{\pre{p}}} = \card{{\post{p}}} = 1$ for any $p \in P$.
$N$ is free-choice net if for any for any $t_1,t_2 \in T$: ${\pre{t_1}} = {\pre{t_2}}$ or ${\pre{t_1}} \cap {\pre{t_2}} = \emptyset$.
$N$ is strongly connected if the graph $(P \cup T,F)$ is strongly connected, i.e., for any two nodes $x$ and $y$ there is a path leading from $x$ to $y$.
\end{definition}

An alternative way to state that a net is free-choice is the requirement that for any $p_1,p_2 \in P$: ${\post{p_1}} = {\post{p_2}}$ or ${\post{p_1}} \cap {\post{p_2}} = \emptyset$.
The Petri net shown in Figure~\ref{fig-not-home} is free-choice. The Petri net shown in Figure~\ref{fig-intro-nfc} is not free-choice because $t4$ and $t5$ shared an input place ($p3$) but have different sets of input places.

\begin{definition}[Siphon and Trap]\label{def:siptrap}
Let $N=(P,T,F)$ be a Petri net and $R \subseteq P$ a subset of places.
$R$ is a siphon if ${\pre{R}} \subseteq {\post{R}}$.
$R$ is a trap if ${\post{R}} \subseteq {\pre{R}}$.
A siphon (trap) is called proper if it is not the empty set.
\end{definition}

Any transition that adds tokens to a siphon also takes tokens from the siphon.
Therefore, an unmarked siphon remains unmarked.
Any transition that takes tokens from a trap also adds tokens to the trap.
Therefore, a marked trap remains marked.

\begin{definition}[Cluster]\label{def:clust}
Let $N=(P,T,F)$ be a Petri net and $x \in P \cup T$.
The cluster of node $x$, denoted $\cluster{x}$ is the smallest set such that
(1) $x \in \cluster{x}$,
(2) if $p \in \cluster{x} \cap P$, then ${\post{p}} \subseteq \cluster{x}$, and
(3) if $t \in \cluster{x} \cap T$, then ${\pre{t}} \subseteq \cluster{x}$.
$\cluster{N}= \{ \cluster{x} \mid x \in P \cup T\}$ is the set of clusters of $N$.
\end{definition}

Note that $\cluster{N}$ partitions the nodes in $N$.
The Petri net shown in Figure~\ref{fig-not-home} has 6 clusters:
$\cluster{N}= \{
\{p1,p2,t1,t2\},\allowbreak
\{p3,t3\},\allowbreak
\{p4,t4\},\allowbreak
\{p5,t5\},\allowbreak
\{p6,t6\},\allowbreak
\{p7,p8,t7\}\}$.

\begin{definition}[Cluster Notations]\label{def:clustnot}
Let $N=(P,T,F)$ be a Petri net and $C \in \cluster{N}$ a cluster.
$P(C) = P \cap C$ are the places in $C$, $T(C) = T \cap C$ are the transitions in $C$, and $M(C) = [p \in P(C)]$ is the smallest marking fully enabling the cluster.
\end{definition}

\begin{definition}[Subnet, P-component, T-Component]\label{def:subcomps}
Let $N=(P,T,F)$ be a Petri net and $X \subseteq P \cup T$ such that $X \neq \emptyset$.
$N\tproj_X = (P \cap X, T \cap X, F \cap (X \times X))$ is the subnet generated by $X$.
$N\tproj_X$ is a \emph{P-component} of $N$ if ${\prenet{N}{p}} \cup {\postnet{N}{p}} \subseteq X$ for $p \in X \cap P$ and $N\tproj_X$ is a strongly connected P-net.
$N\tproj_X$ is a \emph{T-component} of $N$ if ${\prenet{N}{t}} \cup {\prenet{N}{t}} \subseteq X$ for $t \in X \cap T$ and $N\tproj_X$ is a strongly connected T-net.
$\pcomp(N) = \{ X \subseteq P \cup T \mid  N\tproj_X \text{is a P-component} \}$.
$\tcomp(N) = \{ X \subseteq P \cup T \mid  N\tproj_X \text{is a T-component} \}$.
\end{definition}
\begin{figure}[t]
\centering
\includegraphics[width=12cm]{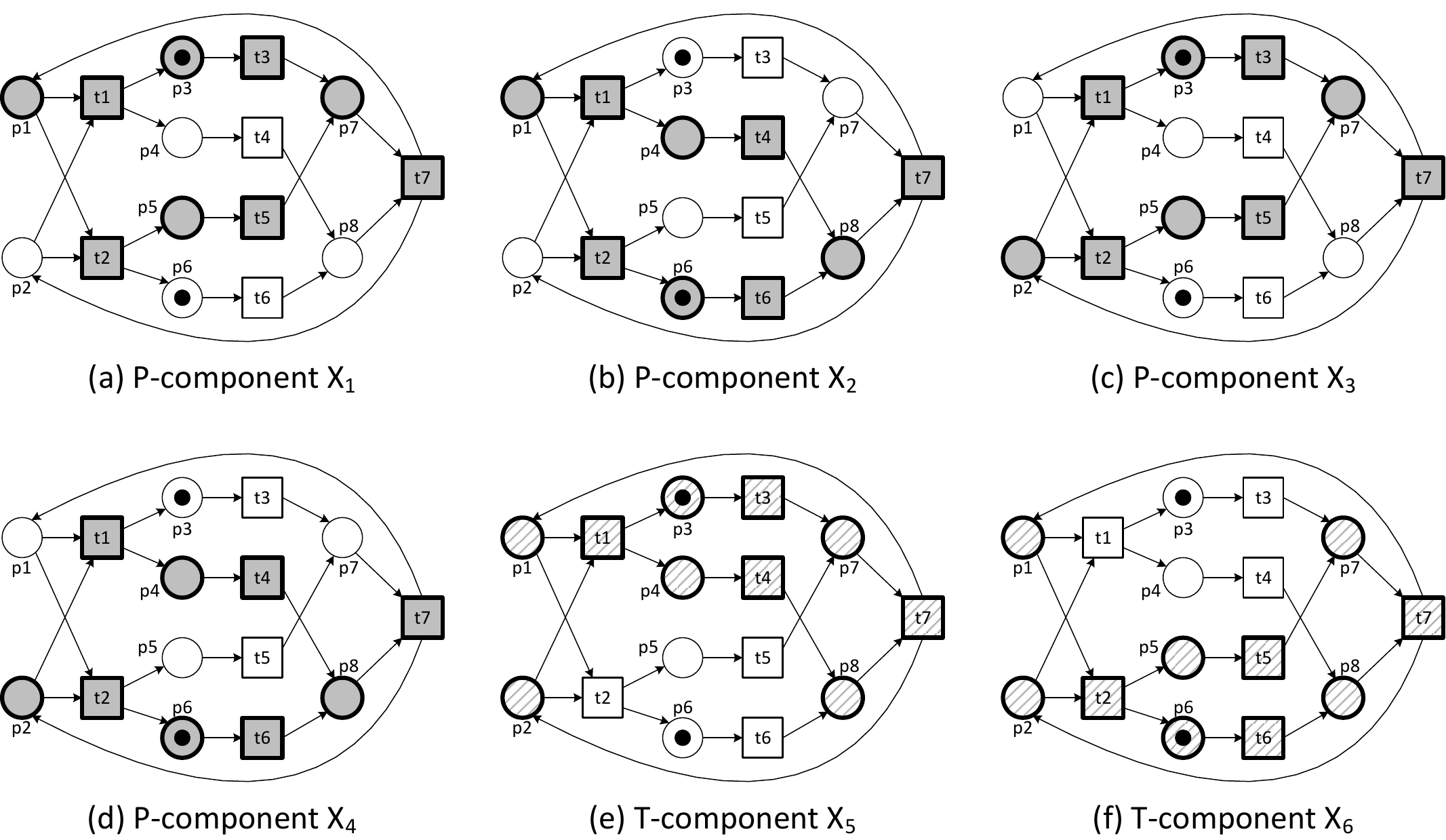}
\caption{The Petri net shown in Figure~\ref{fig-not-home} has four P-components and two T-components.}
\label{fig-not-home-cover}
\end{figure}

The Petri net shown in Figure~\ref{fig-not-home} has four P-components and two T-components (see Figure~\ref{fig-not-home-cover}). These components cover all nodes of the net.

\begin{definition}[P-cover, T-cover]\label{def:covers}
Let $N=(P,T,F)$ be a Petri net.
$N$ has a \emph{P-cover} if and only if $\bigcup \pcomp(N) = P \cup T$.\footnote{$\bigcup Q = \bigcup_{X \in Q} X$ for some set of sets $Q$.}
$N$ has a \emph{T-cover} if and only if $\bigcup \tcomp(N) = P \cup T$.
\end{definition}

Since the early seventies, it is known that well-formed free-choice nets have a P-cover and a T-cover (first shown by Michel Hack).

\begin{theorem}[Coverability Theorem \cite{deselesparza}]\label{theo:cov}
Let $N=(P,T,F)$ be a well-formed free-choice net.
$\bigcup \pcomp(N) = \bigcup \tcomp(N) = P \cup T$.
\end{theorem}

Moreover, for any well-formed free-choice net $N$ and marking $M$: $(N,M)$ is live if and only if every P-component is marked in $M$ (Theorem 5.8 in \cite{deselesparza}).

\subsection{Workflow Nets}
\label{sec:wf-nets}

In the context of business process management, workflow automation, and process mining,
often a subclass of Petri nets is considered where each net has a unique source place $i$ and a unique sink place $o$ \cite{aaljcsc}.

\begin{definition}[Workflow net]\label{def:wf-net}
Let $N=(P,T,F)$ be a Petri net.
$N$ is a workflow net if there are places $i,o \in P$ such that ${\pre{i}} = \emptyset$, ${\post{o}} = \emptyset$, and all nodes $P \cup T$ are on a path from $i$ to $o$.
Given a workflow net $N$, the short-circuited net is $\overline{N}=(P,T \cup \{t^*\},F\cup \{(t^*,i),(o,t^*)\})$.
\end{definition}

The short-circuited net is strongly connected. Different notions of soundness have been defined \cite{soundness-FACS}.
Here we only consider classical soundness \cite{aaljcsc}.

\begin{definition}[Sound]\label{def:sound}
Let $N=(P,T,F)$ be a workflow net with source place $i$. $N$ is sound if and only if $(\overline{N},[i])$ is live and bounded.
\end{definition}

Note that soundness implies that starting from the initial state (just a token in place $i$),
it is always possible to reach the state with one token in place $o$ (marking $[o]$).
Moreover, after a token is put in place $o$, all the other places
need to be empty. Finally, there are no dead transitions (each transition can become enabled).
\begin{figure}[t]
\centering
\includegraphics[width=8.5cm]{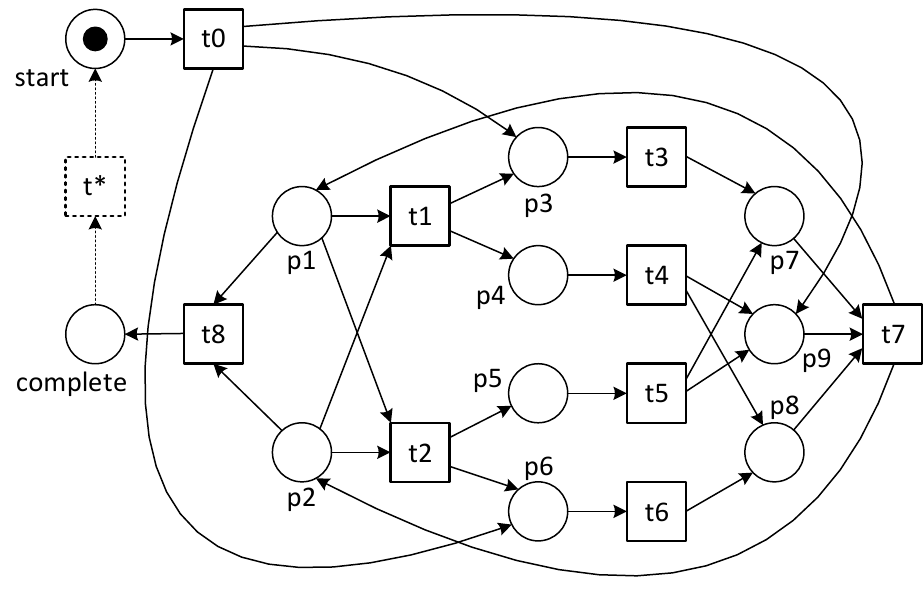}
\caption{The free-choice net without transition $t^*$ is a workflow net.
The short-circuited net is perpetual, i.e., live, bounded, and having a home cluster (e.g., $\{\mi{start},t0\}$).
The short-circuited net is also lucent.}
\label{fig-not-home-wf-net}
\end{figure}

Figure~\ref{fig-not-home-wf-net} shows a sound workflow net. By adding transition $t^*$ the net is short-circuited.
The short-circuited net is live, safe, and cyclic.

\subsection{Uniqueness of Blocking Markings in Free-Choice Nets}
\label{sec:block}

A \emph{blocking marking} is a marking where all transitions in a particular cluster are enabled while all others are disabled.
For example, in Figure~\ref{fig-not-home}, $[p3,p6]$ is not a blocking marking, but $[p3,p8]$, $[p6,p7]$, and $[p7,p8]$ are examples of blocking markings.

\begin{definition}[Blocking Marking]\label{def:bm}
Let $(N,M)$ be a marked net and $C \in \cluster{N}$ a cluster.
A \emph{blocking marking} for $C$ is a marking $M_B\in R(N,M)$ such that $\mi{en}(N,M_B) = T(C)$, i.e.,
all transitions in the cluster are enabled, but no other transitions.
\end{definition}

In \cite{Genrich84} Genrich and Thiagarajan showed that unique blocking markings exist for all clusters in live and safe marked graphs.
This was generalized by Gaujal, Haar, and Mairesse in \cite{blocking-theorem-gaujala2003} where they showed that
blocking markings exist and are unique in live and bounded free-choice nets.
Note that in a free-choice net all transitions in the cluster are enabled simultaneously (or all are disabled).
There is one unique marking in which precisely one cluster is enabled. Moreover, one can reach this marking without firing
transitions from the cluster that needs to become enabled.
A simplified proof was given in \cite{blocking-theorem-wehler2010} and
another proof sketch was provided in \cite{structure-theory-ToPNoC-advanced-course2010}.

\begin{theorem}[Existence and Uniqueness of Blocking Markings \cite{blocking-theorem-gaujala2003}]\label{theo:bm}
Let $(N,M)$ live and bounded free-choice net and $C \in \cluster{N}$ a cluster.
There exists a unique blocking marking for $C$ reachable from $(N,M)$, denoted by $B_{(N,M)}^{C}$.
Moreover, there exists a firing sequence $\sigma \in (T \setminus C)^*$ such that $(N,M)[\sigma\rangle (N,B_{(N,M)}^{C})$.
\end{theorem}
\begin{figure}[thb!]
\centering
\includegraphics[width=7.0cm]{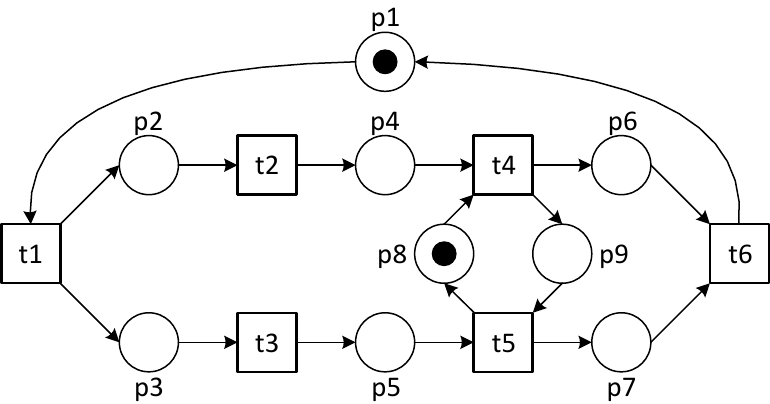}
\caption{A live and safe marked free-choice net that is not locally safe and not perpetual. Nevertheless, the net is lucent.}
\label{fig-not-loc-safe}
\end{figure}

The free-choice net in Figure~\ref{fig-not-loc-safe} is live and bounded.
Hence, each cluster has a unique blocking marking.
The unique blocking marking of the cluster $\{p2,t2\}$ is $[p2,p5]$.
The unique blocking marking of the cluster $\{p6,p7,t6\}$ is $[p6,p7,p8]$.

The free-choice net in Figure~\ref{fig-alternating} has three clusters. The three blocking markings are $[p1,p2]$, $[p1,p3]$, and $[p2,p4]$. Marking $[p3,p4]$ is not a blocking marking because transitions from different clusters are enabled.
\begin{figure}[thb!]
\centering
\includegraphics[width=7.5cm]{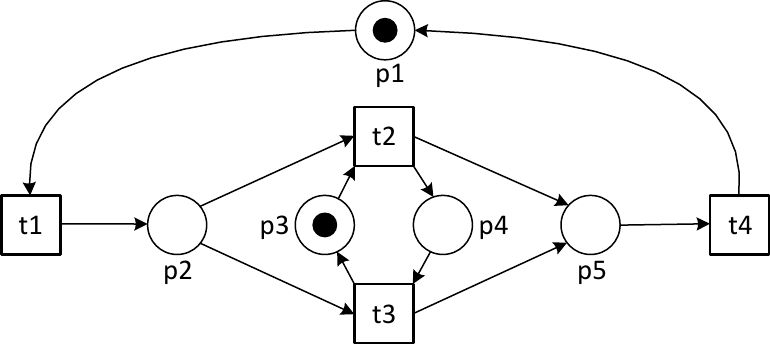}
\caption{A live and locally safe non-free-choice net.
Cluster $\{p1,t1\}$ has two reachable blocking markings $M_1 = [p1,p3]$ and $M_2 = [p1,p4]$.
Also cluster $\{p5,t4\}$ has two reachable blocking markings $M_3 = [p3,p5]$ and $M_4 = [p4,p5]$.}
\label{fig-nfc}
\end{figure}

Figure~\ref{fig-nfc} illustrates that the free-choice property is essential in Theorem~\ref{theo:bm}.
Cluster $C_1 = \{p1,t1\}$ has two reachable blocking markings $M_1 = [p1,p3]$ and $M_2 = [p1,p4]$.
Cluster $C_2 = \{p5,t4\}$ also has two reachable blocking markings $M_3 = [p3,p5]$ and $M_4 = [p4,p5]$.

\section{Lucency}
\label{sec:lucency}

This paper focuses on the question whether markings can be uniquely identified based on the transitions they enable.
Given a marked Petri net we would like to know whether each reachable marking has a unique ``footprint'' in terms of the transitions it enables.
If this is the case, then the net is \emph{lucent}.

\begin{definition}[Lucent]\label{def:lucent}
Let $(N,M)$ be a marked Petri net. $(N,M)$ is lucent if and only if for any $M_1,M_2 \in R(N,M)$: $\mi{en}(N,M_1)=\mi{en}(N,M_2)$ implies $M_1 = M_2$.
\end{definition}

The marked Petri net in Figure~\ref{fig-alternating} is lucent because each of the four reachable markings has a unique footprint in terms of the set of enabled transitions.
The marked Petri net shown in Figure~\ref{fig-intro-nfc} is not lucent because there are two markings $M_1 = [p2,p5]$ and $M_2 = [p2,p6]$ with $\mi{en}(N,M_1) = \mi{en}(N,M_2) = \{t3\}$ and $M_1 \neq M_2$.
The marked Petri nets in figures~\ref{fig-not-home}, \ref{fig-not-home-wf-net}, and \ref{fig-not-loc-safe} are lucent.
The non-free-choice net in Figure~\ref{fig-nfc} is not lucent (markings $[p3,p5]$ and $[p4,p5]$ enable $t4$, and $[p1,p3]$ and $[p1,p4]$ enable $t1$).
Figure~\ref{fig-fc-nonlucid} shows a free-choice net that is also not lucent (markings $[p3,p7,p8]$ and $[p3,p5,p7]$ both enable $\{t1,t4\}$).
\begin{figure}[thbp]
\centering
\includegraphics[width=7.5cm]{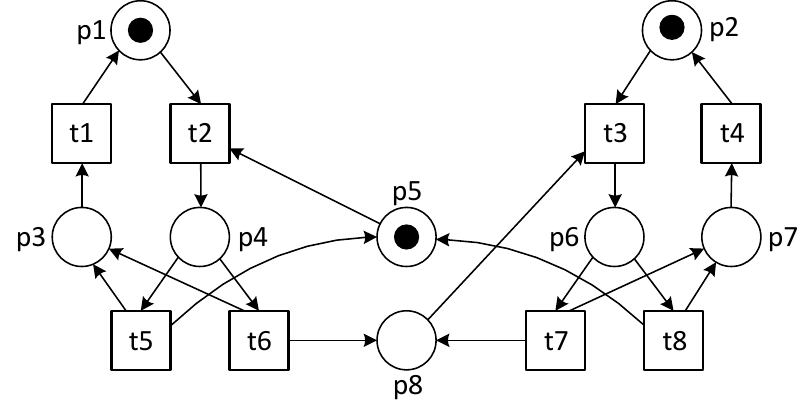}
\caption{A live and locally safe free-choice net that is not lucent because reachable markings $[p3,p7,p8]$ and $[p3,p5,p7]$ both enable $t1$ and $t4$.}
\label{fig-fc-nonlucid}
\end{figure}

\begin{lemma}\label{lemextra}
Let $(N,M)$ be a lucent marked Petri net. $(N,M)$ is bounded and each cluster has at most one blocking marking.
\end{lemma}
\begin{proof}
Assume that $(N,M)$ is both lucent and unbounded. We will show that this leads to a contradiction.
Since $(N,M)$ is unbounded, we can find markings $M_1$ and $M_2$ and sequences $\sigma_0$ and $\sigma$ such that $(N,M) [\sigma_0\rangle (N,M_1) [\sigma\rangle (N,M_2)$ and $M_2$ is strictly larger than $M_1$.
This implies that we can repeatedly execute $\sigma$ getting increasingly larger markings:
$(N,M_2) [\sigma\rangle (N,M_3) [\sigma\rangle (N,M_4) [\sigma\rangle (N,M_5) \ldots$.
At some stage, say at $M_k$, the set of places that are marked stabilizes. However, the number of tokens in some places continues to increase in $M_{k+1}$, $M_{k+2}$, etc.
Hence, we find markings that enable the same set of transitions but that are not the same. For example, $M_{k+1} \neq M_{k+2}$  and $\mi{en}(N,M_{k+1})=\mi{en}(N,M_{k+2})$. Hence, the net cannot be lucent.

Take any cluster $C$ and assume that $(N,M)$ has two different reachable blocking markings $M_1$ and $M_2$.
This means that $\mi{en}(N,M_{1}) = \mi{en}(N,M_{2}) = C \cap T$.
Hence, $(N,M)$ could not be lucent, yielding a contradiction again.
\qed
\end{proof}

We would like to find subclasses of nets that are guaranteed to be lucent based on their structure.
Theorem~\ref{theo:bm} and the fact that lucency implies the existence unique blocking markings,
suggest considering live and bounded free-choice nets.
However, as the example in Figure~\ref{fig-fc-nonlucid} shows, this is not sufficient.

\section{Characterizing Markings of Perpetual Free-Choice Nets}
\label{sec:charmark}

Theorem~\ref{theo:bm} only considers blocking markings, but illustrates that the free-choice property is important for lucency.
Consider for example Figure~\ref{fig-nfc}. $M_1 = [p1,p3]$ and $M_2 = [p1,p4]$ both enable $t1$. $M_3 = [p3,p5]$ and $M_4 = [p4,p5]$ both enable $t4$.
Obviously, the property does not hold for non-free-choice nets even when they are live, safe, cyclic, etc.
However, as Figure~\ref{fig-fc-nonlucid} shows, the property also does not need to hold for free-choice nets even when they are live, safe, and cyclic.
Yet, we are interested in the class of nets for which all reachable markings have a unique ``footprint'' in terms of the transitions they enable. Therefore, we introduce the class of \emph{perpetual nets}. These nets have a ``regeneration point'' involving a so-called ``home cluster''.

\subsection{Perpetual Marked Nets}
\label{sec:def}

A \emph{home cluster} is a cluster corresponding to a home marking, i.e.,
the places of the cluster can be marked over and over again while all places outside the cluster are empty.

\begin{definition}[Home Cluster]\label{def:homeclust}
Let $(N,M)$ be a marked net with $N=(P,T,\allowbreak F)$ and $C \in \cluster{N}$ a cluster of $N$.
$C$ is a \emph{home cluster} of $(N,M)$ if $M(C)$ is a home marking, i.e., for every reachable marking $M' \in R(N,M)$: $M(C) \in R(N,M')$.
\end{definition}

Consider the marked net in Figure~\ref{fig-alternating}. There are three clusters:
$C_1 = \{p1,p2,\allowbreak t1,t2\}$, $C_2 = \{p3,t3\}$, and $C_3 = \{p4,t4\}$.
$C_1$ is a home cluster because $M(C_1)=[p1,p2]$ is a home marking.
$C_2$ is not a home cluster because $M(C_2)=[p3]$ is not a home marking.
$C_3$ is also not a home cluster because $M(C_3)=[p4]$ is not a home marking.

The marked net in Figure~\ref{fig-nfc} also has three clusters: $C_1 = \{p1,t1\}$, $C_2 = \{p2,p3,p4,\allowbreak t2,t3\}$,  and $C_3 = \{p5,t4\}$.
Since $[p1]$, $[p2,p3,p4]$, and $[p5]$ are not home markings, the net has no home clusters.

Nets that are live, bounded, and have at least one home cluster are called \emph{perpetual}.

\begin{definition}[Perpetual Marked Net]\label{def:perpetmn}
A marked net $(N,M)$  is perpetual if it is live, bounded, and has a home cluster.
\end{definition}

The marked Petri nets in figures~\ref{fig-alternating}, \ref{fig-intro-nfc}, \ref{fig-not-home}, and \ref{fig-not-home-wf-net} are perpetual.
The nets in figures~\ref{fig-not-loc-safe}, \ref{fig-nfc}, and \ref{fig-fc-nonlucid} are not perpetual.
Home clusters can be viewed as ``regeneration points'' because the net is always able to revisit a state marking a single cluster.

\begin{lemma}[Sound Workflow Nets are Perpetual]\label{lem}
Let $N$ be a sound workflow net with source place $i$. The short-circuited marked net $(\overline{N},[i])$ is perpetual.
\end{lemma}
\begin{proof}
Soundness implies that $(\overline{N},[i])$ is live and bounded.
Moreover, $[i]$ is a home cluster. It is always possible to enable and fire $t^*$ due to liveness.
After firing $t^*$, place $i$ is marked and there can be no other tokens because otherwise $(\overline{N},[i])$ would be unbounded.
Hence, $[i]$ is a home marking. $\{i\} \cup {\post{i}}$ is a cluster because the transitions in ${\post{i}}$ cannot have additional input places (otherwise they would be dead).
\qed
\end{proof}

Next to workflow nets, there are many classes of nets that have a ``regeneration point'' (i.e., home cluster).
For example, process models discovered by discovery algorithms often have a well-defined start and end point.
By short-circuiting such nets, one gets home clusters.

\subsection{Local Safeness}
\label{sec:locsafe}

It is easy to see that non-safe Petri nets are likely to have different markings enabling the same set of transitions.
In fact, we need a stronger property that holds for perpetual marked free-choice nets: \emph{local safeness}.
Local safeness is the property that each P-component is safe (i.e., the sum of all tokens in the component cannot exceed 1).

\begin{definition}[Locally Safe]\label{lemm:locallysafe}
Let $(N,M)$ be a marked P-coverable net.
$(N,M)$ is \emph{locally safe} if all P-components are safe,
i.e., for any P-component $X \in \pcomp(N)$ and reachable marking $M' \in R(N,M)$: $\sum_{p\in X\cap P} M'(p) \leq 1$.
\end{definition}

Note that a safe marked P-coverable net does not need to be locally safe.
Consider for example the marked net in Figure~\ref{fig-not-loc-safe}.
The net is safe, but the P-component $\{p1,p3,p5,p8,p6,t1,t3,t4,t5,t6\}$ has two tokens.
However, all perpetual marked free-choice nets are locally safe.

\begin{lemma}[Perpetual Marked Free-Choice Nets Are Locally Safe]\label{lemm:hmls}
Let $(N,M)$ be a perpetual marked free-choice net.
$(N,M)$ is locally safe.
\end{lemma}
\begin{proof}
Since $(N,M)$ is perpetual, therefore it is live, bounded, and has a home cluster $C$.
$N$ is well-formed and therefore has a P-cover.
A bounded well-formed free-choice net is only live if every P-component is initially marked (see Theorem 5.8 in \cite{deselesparza}).
Hence, also in home marking $M(C)$ the P-components are marked (the number of tokens is invariant).
Therefore, in any P-component one of the places in $P(C)$ appears. There cannot be two places from cluster $C$ in the same P-component
(this would violate the requirement that transitions in a P-component have precisely one input place).
Hence, each P-component is marked with precisely one token and this number is invariant for all reachable markings.
Hence, $(N,M)$ is locally safe.
\qed
\end{proof}

The nets in figures~\ref{fig-alternating}, \ref{fig-not-home}, and \ref{fig-not-home-wf-net} are free-choice and perpetual and therefore also locally safe.
The net in Figure~\ref{fig-intro-nfc} is locally safe and perpetual, but not free-choice.
The marked Petri net in Figure~\ref{fig-not-loc-safe} is not perpetual and also not locally safe.
Figure~\ref{fig-fc-nonlucid} shows that there are free-choice nets that are live and locally safe, but not perpetual.

\subsection{Realizable Paths}
\label{sec:relpath}

Free-choice nets have many interesting properties showing that the structure reveals information about the behavior of the net \cite{deselesparza}.
Tokens can basically ``decide where to go'', therefore such nets are called free-choice.

The following lemma from \cite{aalcontrolflowacta} shows that tokens can follow an elementary path in a live and bounded free-choice net where the initial marking marks a single place and that is a home marking.

\begin{lemma}[Realizable Paths in Cyclic Free-Choice Nets \cite{aalcontrolflowacta}]\label{lemm:old}
Let $(N,M)$ be a live, bounded, and cyclic marked free-choice net with $M=[p_H]$.
Let $M'$ be a reachable marking which marks place $q$ and let $\langle p_0,t_1,p_1,t_2, \ldots ,t_n,p_n \rangle$
with $p_0=q$ and $p_n=p_H$ be an elementary path in $N$.
There exists a firing sequence $\sigma$ such that
$(N,M')[\sigma\rangle (N,M)$ and each of the transitions $\{t_1, \ldots t_n\}$ is executed in the given order and none
of the intermediate markings marks $p_H$.
\end{lemma}
\begin{proof}
See \cite{aalcontrolflowacta}.
\qed
\end{proof}

Note that Lemma~\ref{lemm:old} refers to a subclass of perpetual marked free-choice nets.
A similar result can be obtained for P-components in a perpetual marked free-choice net.

\begin{lemma}[Realizable Paths Within P-components]\label{lemm:realizablepaths}
Let $(N,M)$ be a perpetual mark-ed free-choice net with home cluster $C$.
Let $X \in \pcomp(N)$ be the nodes of some P-component and $M' \in R(N,M)$ an arbitrary reachable marking.
For any elementary path $\langle p_0,t_1,p_1,t_2, \ldots ,t_n,p_n \rangle \in X^*$ in $N$
with $p_0 \in M'$ and $p_n \in P(C)$:
there exists a firing sequence $\sigma$ such that
$(N,M')[\sigma\rangle (N,M(C))$ and $\sigma \tproj_X = \langle t_1,t_2, \ldots ,t_n \rangle$.
\end{lemma}
\begin{proof}
Let $(P_X,T_X,F_X)$ be the P-component corresponding to $X$.
Note that $p_0$ is the only place of $P_X$ that is marked in $M'$.
Moreover, elements in $\langle p_0,t_1,p_1,t_2,\allowbreak \ldots\allowbreak ,t_n,p_n \rangle \in X^*$ are unique because the path is elementary.
In fact, the places in $\{p_0,p_1, \allowbreak \ldots ,p_n\} \subseteq P_X$ belong to different clusters because a P-component cannot have multiple places of the same cluster.
As a result also $\{t_1,t_2, \ldots t_n\} \subseteq T_X$ belong to different clusters.

If $p_0 = p_n$, then the lemma holds because $p_n \in P(C)$ is marked and we can also mark the other places in $P(C)$.
Theorem~\ref{theo:bm} can be applied such that all places in $P(C)$ can be marked without firing any transitions
in $T(C)$. In fact, there exists a sequence $\sigma_B$ such that $(N,M')[\sigma_B \rangle (N,M(C))$ and $\sigma_B \tproj_X = \langle ~ \rangle$.
$\sigma_B$ does not involve any transitions in $T_X$, because $T(C)$ transitions are not needed and all other transitions in $T_X$
are blocked because $p_n$ is the only place in $P_X$ that is marked.
When all places in $P(C)$ are marked, then all other places need to be empty, otherwise $(N,M)$ is not bounded (see Lemma 2.22 in \cite{deselesparza}).
Hence, $\sigma_B$ leads indeed to $M(C)$.

If $p_0 \neq p_n$, then there is a firing sequence removing the token from $p_0$ (because $M(C)$ is a home marking and $p_0 \not\in M(C)$).
Let $(N,M')[\sigma_1\rangle (N,M_1)$ be the sequence enabling a transition that removes the token from $p_0$ (for the first time).
In $M_1$, transition $t_1 \in {\post{p_0}}$ is enabled (because $N$ is free-choice all transitions in ${\post{p_0}}$ are enabled).
$\sigma_1$ cannot fire any transitions in $T_X$, because $p_0$ is the only place of $P_X$ that is marked.
Therefore, transitions in $\cluster{p_0}$ need to be enabled first.
Let $M_1'$ be the marking after firing $t_1$ ($(N,M')[\sigma_1\rangle (N,M_1)[t_1\rangle (N,M_1')$).
Note that $p_1$ is the only place of $P_X$ marked in $M_1'$.
Let $(N,M')[\sigma_2\rangle (N,M_2)$ be the sequence enabling a transition that removes the token from $p_1$.
Transition $t_2$ is enabled in the marking reached after $\sigma_2$:
$(N,M_2)[t_2\rangle (N,M_2')$. Again $\sigma_2$ cannot involve any transitions in $T_X$ and enables all transitions in ${\post{p_1}}$.
$M_2'$ marks place $p_2$ as the only place in $P_X$.
By recursively applying the argument it is possible to construct the firing sequence
$\sigma' = \sigma_1 \cdot t_1 \cdot \sigma_2 \cdot t_2 \cdot \ldots  \cdot \sigma_n \cdot t_n$ which marks $p_n$.
From the resulting marking one can fire $\sigma_B$ leading to marking $M(C)$.
For the case $p_0 = p_n$ we explained that such a $\sigma_B$ exists.
This shows that we can construct $\sigma = \sigma' \cdot \sigma_B $ such that $(N,M')[\sigma\rangle (N,M(C))$ and $\sigma \tproj_X = \langle t_1,t_2, \ldots ,t_n \rangle$.
\qed
\end{proof}

\subsection{Partial P-covers}
\label{sec:partpcov}

Hack's Coverability Theorem (Theorem~\ref{theo:cov}) states that well-formed free-choice nets have a P-cover.
Our proof that markings are distinguishable based on their enabled transitions exploits this.
In fact, we will construct nets using subsets of P-components. Therefore, we define a notion of a $Q$-projection.

\begin{definition}[Partial P-cover and Projection]\label{def:partialcov}
Let $(N,M)$ be a marked P-coverable Petri net.
Any $Q \subseteq \pcomp(N)$ with $Q\neq \emptyset$ is a partial P-cover of $N$.
$N \tproj_{\bigcup Q}$ is the $Q$-projection of $N$.
$(N \tproj_{\bigcup Q},M\tproj_{\bigcup Q})$ is the marked $Q$-projection of $(N,M)$.
\end{definition}

A $Q$-projection inherits properties from the original net (free-choice and well-formed) and the $Q$-projection is again P-coverable.

\begin{lemma}\label{lemma:partialcovprops}
Let $N=(P,T,F)$ be a P-coverable Petri net, $Q$ a partial P-cover of $N$, and $N \tproj_{\bigcup Q}=(P_Q,T_Q,F_Q)$ the $Q$-projection of $N$.
$\bigcup Q = P_Q \cup T_Q$, $Q \subseteq \pcomp(N \tproj_{\bigcup Q}) \subseteq \pcomp(N)$, and $N \tproj_{\bigcup Q}$ has a P-cover.
\end{lemma}
\begin{proof}
Let $Q= \{ X_1, X_2, \ldots X_n\}$ be P-components of $N$, $P_i = X_i \cap P$, $T_i = X_i \cap T$, for $1 \leq i \leq n$.
$N \tproj_{\bigcup Q}=(P_Q,T_Q,F_Q)$ such that $P_Q = \bigcup_i P_i$ and $T_Q = \bigcup_i T_i$.
Hence, by definition $\bigcup Q = P_Q \cup T_Q$.

Each P-component $X_i$  is fully described by $P_i$, because in any P-component, place $p$ is always connected to the transitions in ${\prenet{N}{p}}$ and ${\postnet{N}{p}}$.
All the original components in $Q$ used to form the partial P-cover of $N$ are also components of $N \tproj_{\bigcup Q}$,
because the subsets of places are in $P_Q$ and all surrounding transitions are included and no new transitions have been added.
However, new combinations may be possible (covering subsets of the places in $P_Q$). Hence, $Q \subseteq \pcomp(N \tproj_{\bigcup Q})$.
$\pcomp(N \tproj_{\bigcup Q}) \subseteq \pcomp(N)$ because a partial P-cover cannot introduce new P-components.
$N \tproj_{\bigcup Q}$ has a P-cover, because $\bigcup \pcomp(N \tproj_{\bigcup Q}) = P_Q \cup T_Q$.
\qed
\end{proof}

\begin{lemma}\label{lemm:well-formed}
Let $N=(P,T,F)$ be a well-formed free-choice net and $Q$ a partial P-cover of $N$.
The $Q$-projection of $N$ (i.e., $N \tproj_{\bigcup Q}$) is a well-formed free-choice net.
\end{lemma}
\begin{proof}
Let $N \tproj_{\bigcup Q} = N_Q = (P_Q,T_Q,F_Q)$.
$N_Q$ is free-choice because $N$ is free-choice and for any added place $p \in P_Q$ all surrounding
transitions ${\pre{p}} \cup {\post{p}}$ are also added. Hence, for any $p_1,p_2 \in P_Q$: ${\post{p_1}} = {\post{p_2}}$ or ${\post{p_1}} \cap {\post{p_2}} = \emptyset$.

$N_Q$ is structurally bounded because it is covered by P-components (Lemma~\ref{lemma:partialcovprops}). The number of tokens in a P-component is constant and serves as an upper bound
for the places in it.

To show that $N_Q$ is structurally live we use Commoner's Theorem \cite{deselesparza}:
``A free-choice marked net is live if and only if every proper siphon includes an initially marked trap''.
Places in $N$ and $N_Q$ have identical pre and post-sets, hence,
for any $R \subseteq P_Q$: ${\prenet{N}{R}} = {\prenet{N_Q}{R}}$ and  ${\postnet{N}{R}} = {\postnet{N_Q}{R}}$.
Hence, $R$ cannot be a siphon in $N$ and not in $N_Q$ (or vice versa).
${\prenet{N}{R}} \subseteq {\postnet{N}{R}}$ if and only if ${\prenet{N_Q}{R}} \subseteq {\postnet{N_Q}{R}}$.
Also, $R$ cannot be a trap in $N$ and not in $N_Q$ (of vice versa).
${\postnet{N}{R}} \subseteq {\prenet{N}{R}}$ if and only if ${\postnet{N_Q}{R}} \subseteq {\prenet{N_Q}{R}}$.
Take any proper siphon $R$ in $N_Q$. This is also a proper siphon in $N$. $R$ contains a proper trap $R'$ in $N$. Clearly, $R' \subseteq P_Q$
and is also a proper trap in $N_Q$. By initially marking all places, $R'$ is also marked and the net is (structurally) live.
Therefore, $N_Q$ is well-formed.
\qed
\end{proof}

A partial P-cover of $N$ may remove places. Removing places can only enable more behavior.
Transitions are only removed if none of the input and output places are included.
Therefore, any firing sequence in the original net that is projected on the set of remaining transitions is enabled in the net based on the partial P-cover.

\begin{lemma}\label{lemm:simu}
Let $(N,M)$ be a live and locally safe marked free-choice net (with $N=(P,T,F)$),
$Q$ a partial P-cover of $N$, and $(N_Q,M_Q)$ the marked $Q$-projection of $(N,M)$ (with $N_Q=(P_Q,T_Q,F_Q)$).
For any sequence $\sigma\in T^*$ that is executable in $(N,M)$
(i.e., $(N,M)[\sigma\rangle(N,M')$), the projected sequence $\sigma_Q = \sigma \tproj_{T_Q}$
is also executable in the marked $Q$-projection and ends in marking $M'\tproj_{\bigcup Q}$
(i.e., $(N_Q,M_Q)[\sigma_Q\rangle \allowbreak (N_Q,M'\tproj_{\bigcup Q})$).
\end{lemma}
\begin{proof}
Let $(N_Q,M_Q)$ be the marked $Q$-projection of $(N,M)$.
$N_Q$ has fewer places. Removing places can only enable more behavior and never block behavior.
Therefore, $\sigma$ is still possible after removing the places not part of any of the included P-components.
After removing these places, transitions not in any included P-component become disconnected and can occur without any constraints.
Hence, $\sigma$ can be replayed and results in the projected marking ($M'\tproj_{\bigcup Q}$).
Removing these transitions from the sequence ($\sigma_Q = \sigma \tproj_{T_Q}$) corresponds to removing them from the net.
Therefore, $(N_Q,M_Q)[\sigma_Q\rangle \allowbreak (N_Q,M'\tproj_{\bigcup Q})$.
\qed
\end{proof}

By combining the above insights, we can show that the $Q$-projection of a perpetual marked free-choice net is again a perpetual marked free-choice net.

\begin{lemma}\label{lemm:perpetualcons}
Let $(N,M)$ be a perpetual marked free-choice net and $Q$ a partial P-cover of $N$.
The marked $Q$-projection of $(N,M)$ (i.e., $(N \tproj_{\bigcup Q},M\tproj_{\bigcup Q})$) is a perpetual marked free-choice net.
\end{lemma}
\begin{proof}
Let $N \tproj_{\bigcup Q} = N_Q = (P_Q,T_Q,F_Q)$ and $M_Q = M\tproj_{\bigcup Q}$.
$N_Q$ is a well-formed free-choice net (see Lemma~\ref{lemm:well-formed}).
To prove that $(N_Q,M_Q)$ is perpetual, we need to show that it is live, bounded, and has a home cluster.

Let $C$ be a home cluster of $(N,M)$. Every P-component of $N$ includes precisely one place of $P(C)$ and holds precisely one token in any reachable state.
Any P-component in $N_Q$ is also a P-component in $N$ (Lemma~\ref{lemma:partialcovprops})
and therefore also has one token in any reachable state. Hence, $(N_Q,M_Q)$ is locally safe.

Every P-component of $N_Q$ is marked in $M$ and also in $M_Q$.
(Lemma~\ref{lemma:partialcovprops} shows that $\pcomp(N \tproj_{\bigcup Q}) \subseteq \pcomp(N)$.
This implies that all components of $N_Q$ are also components of $N$ and thus initially marked.)
Hence, we can apply Theorem 5.8 in \cite{deselesparza} to show that the net is live.

$C_Q = C \cap (P_Q \cup T_Q)$ is a home cluster of $(N_Q,M_Q)$ because the transitions in $C_Q \cap T_Q$ are live and
when they are enabled only the places in $P(C_Q)$ are marked. Hence, $M(C_Q)$ is a home marking.
\qed
\end{proof}

\subsection{Characterization of Markings in Perpetual Free-Choice Nets}
\label{sec:charmain}

We have introduced perpetual free-choice nets because it represents a large and relevant class of models for which
the enabling of transitions uniquely identifies a marking, i.e., these nets are lucent.
In such process models, there can never be two different markings enabling the same set of transitions.
Note that this result is much more general than the blocking marking theorem which only refers to blocking markings and a single cluster.

\begin{theorem}[Characterization of Markings in Perpetual Free-Choice Nets]\label{theo:unmark}
Let $(N,M)$ be a perpetual marked free-choice net.
$(N,M)$ is lucent.\footnote{The original proof published in ``Wil M. P. van der Aalst:
Markings in Perpetual Free-Choice Nets Are Fully Characterized by Their Enabled Transitions. Petri Nets 2018: 315-336'' contains an error that can be repaired, but this makes the proof overly complex. One needs to consider a sequence of disagreeing P-components. Here, a simpler, and more direct, proof is given that does not require some of the intermediate results presented before. The earlier results are still valid and meaningful. The proof now looks more involved, because it is self-contained and does not use any intermediate results. Interestingly, the proof can be extended beyond live free-choice nets, but this is out of scope for this corrected proof.}
\end{theorem}
\begin{proof}
Let $(N,M)$ be a perpetual marked free-choice net, i.e., $(N,M)$ is live, bounded, and has a home cluster $C$.
For any $M_1,M_2 \in R(N,M)$ such that $\mi{en}(N,M_1)=\mi{en}(N,M_2)$, we need to prove that $M_1 = M_2$.
We assume that this is not the case, and show that this leads to a contradiction.
\def\1token{\raisebox{.5pt}{\textcircled{\raisebox{-.9pt} {1}}}}
\def\2token{\raisebox{.5pt}{\textcircled{\raisebox{-.9pt} {2}}}}
\def\token{\raisebox{.5pt}{\textcircled{\raisebox{-.3pt} {$\bullet$}}}}

Let $M_1^c$ and $M_2^c$ be such that $\mi{en}(N,M_1^c)=\mi{en}(N,M_2^c)$ and $M_1^c \neq M_2^c$. Let us consider the tokens in both markings and partition these into three classes: 
$\token$-tokens,
$\1token$-tokens, and 
$\2token$-tokens.
$\token$-tokens are tokens where $M_1^c$ and $M_2^c$ agree. These can be found by taking the maximal submarking contained in both.
$\token$-tokens are also called ``agreement tokens''.
$\1token$-tokens are the tokens in $M_1^c$, but not $M_2^c$.
$\2token$-tokens are the tokens in $M_2^c$, but not $M_1^c$.
Hence, $M_1^c$ is composed of $\token$-tokens and $\1token$-tokens, and $M_2^c$ is composed of $\token$-tokens and $\2token$-tokens. 
$\1token$- and $\2token$-tokens are also called ``disagreement tokens''.
We will use this notation throughout the proof.
Because $(N,M)$ is safe, each place can have only one of these
$\token$-,
$\1token$-,
$\2token$-tokens.
For example, if a place contains both a 
$\1token$-token and 
$\2token$-token, then it has a $\token$-token instead.

Starting from both $M_1^c$ and $M_2^c$, we synchronously execute transitions using only $\token$-tokens. Since we are not using
$\1token$ or $\2token$-tokens, we can do this synchronously.
In this process, each place remains safe (at most one $\token$-,
$\1token$-,
$\2token$-token). If we try to move the $\token$-tokens closer to 
the home cluster $C$, there is a point at which this is no longer possible (use liveness and the free-choice property, and move towards the home marking $M(C)$).
Hence, we end up in two new markings $M_1$ and $M_2$, where
all enabled transitions in $M_1$ need to consume both 
$\token$- and
$\1token$-tokens, and 
all enabled transitions in $M_2$ need to consume both 
$\token$- and
$\2token$-tokens.
In none of the markings, a transition is enabled based on only ``agreement tokens'' or only ``disagreement tokens''.
Compared to $M_1^c$ and $M_2^c$ only $\token$-tokens were moved. 

In the remainder, we only consider the markings $M_1$ and $M_2$ just constructed.
In both of these markings, only clusters containing ``agreement tokens'' and ``disagreement tokens'' are enabled, i.e., an enabled cluster has either $\token$- and
$\1token$-tokens (enabled in $M_1$), or $\token$- and
$\2token$-tokens (enabled in $M_2$).

Pick an arbitrary cluster enabled in marking $M_1$. We refer to this cluster as $C_1$. 
All places in this cluster are marked with $\token$- and
$\1token$-tokens, but the cluster is not enabled in $M_2$.
However, there must be a firing sequence starting in $M_2$, putting a token in one of the places in $C_1$.
Take a \emph{shortest} firing sequence $\sigma_s$ starting in $M_2$ and putting a token in one of the empty places in $C_1$.
This sequence must start with an enabled cluster having only $\token$- and
$\2token$-tokens. We refer to this cluster as $C_2$. 
Note that $C$, $C_1$, and $C_2$ should be different.
($C_1$ and $C_2$ have ``disagreement tokens'' and cannot be the same.
$C$ cannot be the same as $C_1$, because $C$ would be enabled in $M_1$
while having still tokens in $C_2$. Etc.)

Let $p^{mrk}$ be a place in $C_2$ having a $\token$-token.
Based on $\sigma_s$, we can construct a path $\rho_s$ in the Petri net
starting in $C_2$ leading to one of the unmarked places in $C_1$.
This path $\rho_s$ ``follows a token'' from $C_2$ to $C_1$ using only transitions in $\sigma_s$. Just ``color'' the token in $p^{mrk}$ and all it descendants to see that such a path exists (also note that to reach $M(C)$ the agreement tokens in $C_1$ need to be removed).
Next, we ``compact'' path $\rho_s$ into another path $\rho_1$ which visits each cluster only once.
This is possible because if $\rho_s$ visits a cluster multiple times, we can create a short-cut and immediately jump to the last occurrence of the cluster (in a cluster all places are connected to all transitions due to the free-choice property).

To explain the mechanism of compacting the path, assume
$\rho_s = \langle p^s_0,t^s_1,p^s_1, \allowbreak \ldots, \allowbreak p^s_{k-1},\allowbreak t^s_{k},p^s_k
 \rangle$ with $p^s_0 = p^{mrk}$ and $p^s_k \in C_1$.
We inspect the path from left to right using a pointer $i$ with $0 \leq i \leq k$.
Initially, $i=0$.
If place $p^s_i$ is an element of $C_1$, 
then we remove the rest of $\rho_s$ because we have already reached the target cluster $C_1$.
If place $p^s_i$ is not in $C_1$, but there is a later place $p^s_j$ belonging to the same cluster, then we remove the subsequence $\langle t^s_{i+1},p^s_{i+1}, \ldots , t^s_{j},p^s_j
 \rangle$ from $\rho_s$ and set $i$ to $j$.
In this case, we take $j$ to be maximal, i.e., the latest visit of $\rho_s$ to the cluster. (Note that $p^s_i$ is also an input place of $t^s_{j+1}$, because both belong to the same cluster. Hence, the path remains connected.)
If place $p^s_i$ is the only place of a cluster on the path, 
then we keep $p^s_i$ and $t^s_{i+1}$ and set $i$ to $i+1$ (i.e., move to the next place).
This process can be repeated until we reach $C_1$.
The resulting path is $\rho_1$ and is a subsequence of $\rho_s$, but still a path of $N$.

Using the above construction, we can construct a path 
$\rho_1 = \langle p^{mrk},t_1,p_1, \ldots , p_{n-1},\allowbreak t_{n},p^{conn}
 \rangle$ leading from 
$p^{mrk} \in C_2$ to some place $p^{conn} \in C_1$ such that the path contains a subset of the transitions in $\sigma_s$
and does not visit the same cluster twice, i.e.,
each place in $P_1 = \{ p^{mrk},p_1, \ldots , p_{n-1},p^{conn}\}$
corresponds to a unique cluster.

There is also a path from $p^{conn}$ to the home cluster $C$ because the net is strongly connected. Intuitively, we can also ``follow a token'' from $C_1$ to $C$.  Also this path can be compacted into another path $\rho_2$ which starts $p^{conn}$ and ends in some place $p^{end} \in C$ and visits each cluster only once.
We use the same principle as before: if a cluster is visited multiple times along the path, we jump immediately to the last occurrence
(using the fact that in a cluster all places are connected to all transitions).
Therefore, we can construct a path 
$\rho_2 = \langle p^{conn},t_{n+1},p_{n+1}, \ldots , p_{m-1},t_{m},p^{end}
 \rangle$ leading from 
$p^{conn} \in C_1$ to some place $p^{end} \in C$
visiting each cluster at most once, i.e.,
each place in $P_2 = \{ p^{conn},p_{n+1}, \ldots , p_{m-1},p^{end}\}$
corresponds to a unique cluster.

Next to 
$\rho_1 = \langle p^{mrk},t_1,p_1, \ldots , p_{n-1},t_{n},p^{conn}
 \rangle$ and
$\rho_2 = \langle p^{conn},t_{n+1},p_{n+1}, \ldots\allowbreak , p_{m-1},t_{m},p^{end}
 \rangle$, we create two more paths:
 $\rho_{1+2} = \langle p^{mrk},t_1,p_1, \ldots , p_{n-1},t_{n},p^{conn},\allowbreak t_{n+1},p_{n+1}, \ldots , p_{m-1},t_{m},p^{end}
 \rangle$
 and 
${\rho_{2}}' = \langle p^{alt},t_{n+1},p_{n+1}, \ldots , p_{m-1},t_{m},p^{end}
 \rangle$ with $p^{alt} \in C_1$ having a $\token$-token (i.e., marked in both $M_1$ and $M_2$).
Also in the paths $\rho_{1+2}$ and ${\rho_{2}}'$, 
each cluster appears at most once. For ${\rho_{2}}'$ this is obvious, because 
$p^{alt}$ and $p^{conn}$ are in the same cluster and 
for the rest ${\rho_{2}}'$ and $\rho_2$ are identical.

To see that also in $\rho_{1+2}$ each cluster appears at most once, assume there is a cluster $C'$ that appears at least twice in $\rho_{1+2}$.
This implies there is a place $p_1' \in P_1 \cap C'$ and a place $p_2' \in P_2 \cap C'$. If one of these places is $p^{conn}$, 
then this is impossible, because cluster $C'$ did not appear twice in the shorter sequences $\rho_1$ and $\rho_2$. Moreover, also when both are different from $p^{conn}$, this is still not possible. 
Consider a transition $t' \in \sigma_s$ that consumed a token from $p_1'$ while executing $\sigma_s$. Due to the free-choice property, transition $t'$ also consumed a token from $p_2'$.
Transition $t'$ occurred before the cluster $C_1$ got enabled, i.e., there was still a $\token$-token in $p^{alt}$.
However, this implies that there is a marking $M'$ in which 
the places $p_1'$, $p_2'$, and $p^{alt}$ are marked at the same time.

Now consider path ${\rho_{2}}' = \langle p^{alt},t_{n+1},p_{n+1}, \ldots , p_{m-1},t_{m},p^{end}
 \rangle$ in the marking $M'$ just described. 
At least two places on this path contain a token ($p_2'$ and $p^{alt}$). Starting from this marking $M'$, move these two tokens along the path ${\rho_{2}}'$ towards place $p^{end}$.
Because the net is free-choice, we can control the clusters visited on this path and each transition in one of these clusters consumes precisely one token from path ${\rho_{2}}'$, and produces at least one token for path ${\rho_{2}}'$. Hence, the number of tokens on the path does not decrease. Since the net has a home cluster containing $p^{end}$, we can continue to do this until there are multiple tokens in $p^{end}$ or until one of the transitions in $C$ fires. If a transition on the path is enabled, we fire it. If no transition on the path is enabled
and we did not reach the home marking $M(C)$ yet, we can execute transitions towards the home marking $M(C)$.
By following this process, we can put two tokens in $p^{end}$
or fire a transition from cluster $C$ while there are still tokens in the rest of the net.
This leads to a contradiction (the net is safe and cannot have two tokens in a place or reach markings strictly larger than $M(C)$).
Therefore, also $\rho_{1+2}$ has the property that each cluster is visited at most once (like
$\rho_1$, $\rho_2$, and ${\rho_{2}}'$).

Now consider the status of path $\rho_{1+2} = \langle p^{mrk},t_1,p_1, \ldots , p_{n-1},t_{n},p^{conn},\allowbreak t_{n+1},\allowbreak p_{n+1},\allowbreak \ldots , p_{m-1},t_{m},p^{end}
 \rangle$ (which visits clusters at most once) in marking $M_1$.
In this marking, $p^{mrk}$ and $p^{conn}$ are marked.
$p^{mrk}$ has a $\token$-token and $p^{conn}$ a $\1token$-token.
Hence, at least two places on path $\rho_{1+2}$ contain a token
and all places on the path belong to different clusters.
Therefore, we can use again the strategy to move tokens along the path towards place $p^{end}$.
Each transition in $\rho_{1+2}$ consumes precisely one token from the path and produces at least one token for the path.
As long as none of the transitions in $C$ occurs, we can ensure that the number of tokens on the path does not decrease
and move tokens towards $p^{end}$. This leads to a contradiction the moment there are two tokens in $p^{end}$ or the transitions in $C$ are enabled while having additional tokens (the net is safe).
Since $\rho_{1+2}$ cannot exist, also $M_1$ and $M_2$ cannot exist.

$M_1$ and $M_2$ were constructed using $M_1^c$ and $M_2^c$.
Hence, there cannot be $M_1^c$ and $M_2^c$ be such that $\mi{en}(N,M_1^c)=\mi{en}(N,M_2^c)$ and $M_1^c \neq M_2^c$.
Hence, the initial assumption leads to a contradiction, showing that the net must be lucent.
\qed
\end{proof}
\begin{figure}[t]
\centering
\includegraphics[width=7.5cm]{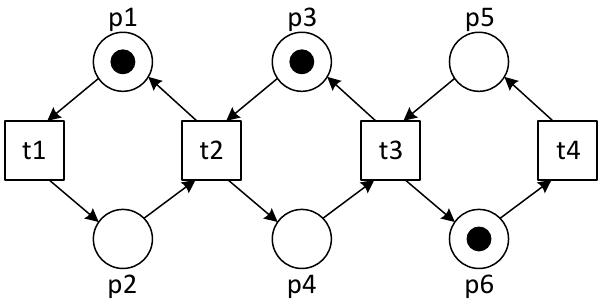}
\caption{A live and locally safe marked free-choice net that is not perpetual.
The model is also not lucent since there are two reachable markings $M_1 = [p1,p3,p6]$ and $M_2 = [p1,p4,p6]$ that both enable $t1$ and $t4$.}
\label{fig-locally-safe-not-perpetual}
\end{figure}

For the class of perpetual marked free-choice nets, markings are uniquely identified by the set of enabled transitions.
As shown before, the free-choice property is needed and liveness and boundedness are not sufficient.
The above theorem also does not hold for live and locally safe marked free-choice nets (see for example Figure~\ref{fig-fc-nonlucid}).
The requirement that the net has a home cluster seems essential for characterizing marking in terms of enabled transitions.

Consider for example the live and locally safe marked free-choice net in Figure~\ref{fig-locally-safe-not-perpetual}.
There are three P-components: $\{p1,p2,t1,t2\}$, $\{p3,p4,t2,t3\}$, and $\{p5,p6,t3,t4\}$.
These always contain precisely one token. However, there are two reachable markings $M_1 = [p1,p3,p6]$ and $M_2 = [p1,p4,p6]$ that both enable $t1$ and $t4$.
This can be explained by the fact that the net is not perpetual.
There are four clusters, but none of these clusters is a home cluster.
Note that the counter-example in Figure~\ref{fig-locally-safe-not-perpetual} is actually a marked graph.
This illustrates that the home cluster requirement is also essential for subclasses of free-choice nets.

\section{Conclusion and Implications}
\label{sec:concl}

We started this paper by posing the question: ``What is the class of Petri nets for which the marking is uniquely identified by the set of enabled transitions?''.
This led to the definition of \emph{lucency}.
The main theorem proves that markings from perpetual marked free-choice nets are guaranteed to be lucent.
Moreover, we showed that all requirements are needed (in the sense that dropping any of the requirements yields a counterexample).
Table~\ref{tab:results} provides an overview of the examples used in this paper.
For example, even live and safe free-choice nets may have multiple markings having the same set of enabled transitions.
\begin{table}[t]
\centering
\caption{Overview of the examples used:
\emph{Marked PN} = figure showing a marked Petri net,
\emph{FreC} = free-choice, \emph{Live} = live, \emph{Boun} = bounded, \emph{Safe} = safe, \emph{LocS} = locally safe, \emph{PC} = number of P-components,
\emph{HClu} = net has at least one home cluster, \emph{Perp} = perpetual, \emph{UnBM} = net has unique blocking marking for each cluster,
\emph{Lucent} = lucent, \emph{Pls} = number of places, \emph{Trs} = number of transitions, and \emph{RM} = number of reachable markings.}\label{tab:results}
\begin{tabular}{|c|c|c|c|c|c|c|c|c|c|c|c|c|c|}
  \hline
\emph{Marked PN} & \emph{FreC} & \emph{Live} & \emph{Boun} & \emph{Safe} & \emph{LocS} & \emph{PC} & \emph{HClu} & \emph{Perp} & \emph{UnBM} & \emph{Lucent} & \emph{Pls} & \emph{Trs} & \emph{RM}\\   \hline
Figure~\ref{fig-alternating} & Yes & Yes & Yes & Yes & Yes & 2 & Yes & Yes & Yes & Yes & 4 & 4 & 4\\
Figure~\ref{fig-intro-nfc} & No & Yes & Yes & Yes & Yes & 2 & Yes & Yes & No & No & 6 & 6 & 6\\
Figure~\ref{fig-not-home} & Yes & Yes & Yes & Yes & Yes & 4 & Yes & Yes & Yes & Yes & 8 & 7 & 9\\
Figure~\ref{fig-not-home-wf-net} & Yes & Yes & Yes & Yes & Yes & 6 & Yes & Yes & Yes & Yes & 11 & 10 & 11\\
Figure~\ref{fig-not-loc-safe} & Yes & Yes & Yes & Yes & No & 5 & No & No & Yes & Yes & 9 & 6 & 8\\
Figure~\ref{fig-nfc} & No & Yes & Yes & Yes & Yes & 2 & No & No & No & No & 5 & 4 & 6\\
Figure~\ref{fig-fc-nonlucid} & Yes & Yes & Yes & Yes & Yes & 3 & No & No & Yes & No & 8 & 8 & 12\\
Figure~\ref{fig-locally-safe-not-perpetual} & Yes & Yes & Yes & Yes & Yes & 3 & No & No & Yes & No & 6 & 4 & 8\\
  \hline
\end{tabular}
\end{table}

Other characterizations may be possible. An obvious candidate is the class of Petri nets without PT and TP handles \cite{handles}.
As shown in \cite{aalwfm98lncs} there are many similarities between free-choice workflow nets and well-structured (no PT and TP handles)
workflow nets when considering notions like soundness and P-coverability. Moreover, it seems possible to relax the notion of a regeneration point by considering simultaneously marked clusters.

Structure theory aims to link structural properties of the Petri net to its behavior.
The connection between lucency and home clusters in free-choice nets could be relevant for verification and synthesis problems.
The ability to link the enabling of transitions to states (i.e., lucency) is particularly relevant when observing running systems or processes, e.g.,
in the field of process mining \cite{process-mining-book-2016} where people study the relationship between modeled behavior and observed behavior.
If we assume lucency, two interesting scenarios can be considered:
\begin{itemize}
  \item \emph{Scenario 1: The system's interface or the event log reveals the set of enabled actions.} At any point in time or for any event in the log, we know the internal state of the system or process. This makes it trivial to create an accurate process model (provided that all states have been visited).
  \item \emph{Scenario 2: The system's interface or the event log only reveals the executed actions.} The internal state of the system is unknown, but we know that it is fully determined by the set of enabled actions (some of which may have been observed).
\end{itemize}
It is easy to create a discovery algorithm for the first scenario.
The second scenario is more challenging. However, the search space can be reduced considerably by assuming lucency (e.g., learning perpetual marked free-choice nets).
Hence, Theorem~\ref{theo:unmark} may lead to new process mining algorithms or help to prove the correctness and/or guarantees of existing algorithms.

Assume that each event in the event log is characterized by $e=(\sigma_{\mi{pref}},a,\sigma_{\mi{post}})$ where $\sigma_{\mi{pref}}$ is the prefix (activities that happened before $e$), $a$ is the activity executed,
 and $\sigma_{\mi{post}}$ is the postfix (activities that happened after $e$). The result of applying a process discovery algorithm can be seen as a function $\mi{state}()$ which maps any event $e$ onto a state $\mi{state}(e)$,
 i.e., the state in which $e$ occurred (see \cite{wires-replay,two-step-mining-SOSYM} for explanations). Hence, events $e_1$ and $e_2$ satisfying $\mi{state}(e_1)=\mi{state}(e_2)$ occurred in the same state and can be viewed as ``equivalent''. This way discovery is reduced to finding an \emph{equivalence relation} on the set of events in the log. Given such an equivalence relation, we can apply the approach described under Scenario 1.
 Viewing process discovery as ``finding an equivalence relation on events'' provides an original angle on this challenging and highly relevant learning task.

\bibliographystyle{plain}
\bibliography{lit}

\end{document}